%% file: QmsStabilityRevision.tex
\documentclass[a4paper,reqno]{amsart}

\usepackage[centertags]{amsmath}
\usepackage{amsfonts}
\usepackage{euscript}
\usepackage{amssymb}
\usepackage{amsthm}
\usepackage{newlfont}
\usepackage{euscript}
\usepackage{graphicx}
\usepackage{subfigure}
\usepackage{caption}
\usepackage[usenames,dvipsnames]{color}
\usepackage{enumerate}
\theoremstyle{plain}
  \newtheorem{theorem}{Theorem}[section]
  \newtheorem{corollary}[theorem]{Corollary}
  
  \newtheorem{lemma}[theorem]{Lemma}
\theoremstyle{definition}

\theoremstyle{remark}
  \newtheorem{remark}[theorem]{Remark}

\numberwithin{equation}{section}

 \let\Ga=\Gamma \let\La=\Lambda

\newcommand{\caA}{{\mathcal A}}
\newcommand{\caB}{{\mathcal B}}

\newcommand{\caG}{{\mathcal G}}
\newcommand{\caH}{{\mathcal H}}

\newcommand{\caK}{{\mathcal K}}
\newcommand{\caL}{{\mathcal L}}

\newcommand{\caS}{{\mathcal S}}

\newcommand{\caV}{{\mathcal V}}
\newcommand{\caW}{{\mathcal W}}

\newcommand{\bbC}{{\mathbb C}}

\newcommand{\bbR}{{\mathbb R}}

\newcommand{\bbZ}{{\mathbb Z}}

\newcommand{\bsE}{{\boldsymbol E}}

\DeclareMathAlphabet{\mathpzc}{OT1}{pzc}{m}{it}

\newcommand{\lVERT}{\lvert\mkern-2mu\lvert\mkern-2mu \lvert}
\newcommand{\rVERT}{\rvert\mkern-2mu\rvert\mkern-2mu \rvert}

\DeclareMathOperator{\Tr}{Tr}

\newcommand{\beq}{\begin{equation}}
\newcommand{\eeq}{\end{equation}}
\newcommand{\baq}{\begin{eqnarray}}
\newcommand{\eaq}{\end{eqnarray}}

\newcommand{\e}{\mathrm{e}}

\newcommand{\opunit}{\text{1}\kern-0.22em\text{l}}

\allowdisplaybreaks

\begin{document}

\title{Uniqueness regime for Markov dynamics on quantum lattice spin systems}
\author{N.~Crawford$^1$}
\email[N.~Crawford]{nickc@tx.technion.ac.il}
\author{W.~De Roeck$^2$}
\email[W.~De Roeck]{wojciech.deroeck@fys.kuleuven.be}
\author{M.~Sch\"utz$^2$}
\email[M.~Sch\"utz]{marius.schutz@fys.kuleuven.be}

\maketitle
\markboth{N.~Crawford, W.~De Roeck, and M.~Sch\"utz}{Uniqueness regime for quantum Markov dynamics}
\vspace{-11pt}
 \begin{center} {\footnotesize  
    \textit{$^1$ Mathematics Department, The Technion, Haifa (Israel).}
 
 \textit{$^2$ Institute for Theoretical Physics, KU Leuven, Leuven (Belgium).} }\end{center}

\begin{abstract}
We consider a lattice of weakly interacting quantum Markov processes. Without interaction, the dynamics at each site is relaxing exponentially to a unique stationary state. 
With interaction, we show that there remains a unique stationary state in the thermodynamic limit, i.e.\ absence of phase coexistence, and the relaxation towards it is exponentially fast for local observables. We do not assume that the quantum Markov process is reversible (detailed balance) w.r.t.\ a local Hamiltonian. 
\end{abstract}

\section{Introduction}

Stochastic Markov dynamics is a crucial modeling device for (non-equilibrium) statistical mechanics, as well as a technical tool for simulation purposes, i.e.\ Monte Carlo algorithms.
In the past years, Markov evolutions have been also considered for many-body quantum systems, see e.g.~\cite{AFH,AHHH,brandao,KBDKMZ,Mat,MOZ,P,TOVP,TWV,ZBC}.
A lot remains to be clarified about their physical origin and justification,
but  in the present paper  we take a top-down approach and we assume that a Markovian dynamics is given. Its generator (Lindblad operator) is a sum of quasi-local terms.  We allow this dynamics to be non-equilibrium in that it does not need to satisfy a detailed balance condition, i.e.\ we do not ask that it is self-adjoint w.r.t.\ to a KMS state (Gibbs state). We are motivated by the following question: When does the Markov dynamics allow for phase coexistence, i.e.\ multiple stationary states in the thermodynamic limit?
A property that appears often together with uniqueness, i.e.\ no phase coexistence,  is (uniform) exponential relaxation to the stationary state, meaning that expected values of local observables approach their stationary values exponentially fast independent of the volume and boundary conditions. This notion and the related property of `rapid mixing', which was important in recent work \cite{BCLMP,CLMP,KE} on similar questions as ours, is made more precise in the main part of the text. 
 
As such, the question of phase uniqueness is also interesting for classical systems.  One could conjecture that for any detailed balance dynamics in the uniqueness regime (e.g.~at high enough temperature for Glauber type dynamics, etc.) the uniqueness is stable against small non-detailed balance perturbation.  To our knowledge, this has yet been proven in any generality. Nevertheless, there are numerous results on uniqueness and uniform exponential relaxation (and log-Sobolev inequalities), in particular in the regime of `complete analyticity'.  We refer to works like \cite{guionnetzegarlinski,MS, Mart,zegarlinskigeneral} and references therein.

On the quantum side,  uniqueness and exponential approach have been established in \cite{MOZ,MZ} for small perturbations of product dynamics with specific generators of the form $E_X-\opunit$, where $E_X$ maps arbitrary operators into operators that act as (multiple of) the identity within the spatial set $X$. 
In \cite{KT}, exponential relaxation is proven starting from a log-Sobolev inequality, and in \cite{RW}, exponential relaxation was obtained for a class of cellular automata.

Our work is concerned with general perturbations of product dynamics (a classical analogue was treated in \cite{netocnymaesergodicity}).  We show that the perturbed dynamics remains exponentially ergodic with a unique (infinite-volume) stationary state, see Theorem \ref{thm}. Additionally, if the dynamics depends analytically on some parameter, then so does the stationary state. Boundary conditions do not affect the thermodynamic limit of the dynamics and stationary state, which is the content of Corollary \ref{cor}.

\subsection*{Acknowledgements} It is a pleasure to thank Christian Maes for interesting discussions. WDR and MS are thankful for financial support from the DFG (German Research Foundation). NC is supported in part by ISF grant No. 915/12.

\section{Definition of the Dynamics}

\subsubsection*{Quantum Spin System}
We consider quantum spin systems on the the $d$-dimensional lattice $\bbZ^d$, $d\geq 1$. The lattice distance between two sites $x,y\in\bbZ^d$ is denoted by $d(x,y)$ and we use the same notation for the minimal distance between two sets of sites. For every site $x\in\bbZ^d$ let $\caH_{x}$ be a copy of a finite dimensional Hilbert space $\caH$. The operators $\caA_x:=\caB(\caH_{x})$ equipped with operator norm $\lVert\,\cdot\,\rVert$ constitute the algebra of single-site observables. 
We use the symbol $\Subset$ to indicate finite subsets.
The physical Hilbert space and the algebra of observables for our quantum spin system in a volume $\Lambda\Subset \bbZ^d$ are then given by the tensor products 
\begin{equation}
 \caH_\La := \textstyle{\bigotimes_{x\in\La}} \caH_{x}, \qquad \caA_\La:=\textstyle{\bigotimes_{x\in\La}} \caA_{x}.
\end{equation} 
For $S\subset\La$ each $A\in\caA_{S}$ can be embedded canonically in $\caA_\La$ as the local observable $A\otimes \opunit$, which, we say, has support in $S$. In the same way also each (super-)operator $O\in\caB(\caA_{S})$ will be identified with a local operator on $\caB(\caA_{\La})\cong \caB(\caA_{S})\otimes\caB(\caA_{\Lambda\setminus S})$ without explicitly mentioning it. The quasi-local algebra $\caA$ is the $C^*$ algebra defined as the operator norm closure of the inductive limit $\bigcup_\La \caA_\La$ with subsets $\La\Subset\bbZ^d$ ordered by inclusion. A state is a normalized positive functional on a $C^*$ algebra. The operator norm on $\caB(\caA_{\La})$ 
is again denoted by $\lVert\,\cdot\, \rVert$. 
Note that the norm of a local operator $O\in\caB(\caA_\Lambda)$ with support in $S\subset\Lambda$ may increase with the ambient volume $\Lambda\Subset\bbZ^d$, which motivates the definition of the completely bounded norm $\lVert O \rVert_{\mathrm{cb}}:=\sup_\Lambda \lVert O\rVert$ on the space of local operators\footnote{Since however $\lVert O\rVert_{\mathrm{cb}}\leq \dim(\caH_S) \lVert O\rVert$, where the right hand side means the norm in $\caB(\caH_S)$, which is shown for example in the textbook \cite{Pa} (Proposition 8.11), we use the the completely bounded norm mostly to facilitate volume-independent notation.}.

\subsubsection*{Unperturbed Product Dynamics}
For each $x\in\bbZ^d$ let $\caG(x)\in\caB(\caA_{x})$ be the generator of a semigroup $\e^{t \caG(x)}$, $t\geq 0$, of completely positive and identity preserving operators on $\caA_{x}$, i.e.~a quantum Markov semigroup (QMS) in the Heisenberg picture. 

Our \emph{main assumption} is that $0$ is a simple eigenvalue of each $\caG(x)$ with corresponding rank one projection $Q_x$ and that there exist constants $M,g>0$ such that 
\begin{equation}\label{eq: AssDecay}
 \sup_{x\in\bbZ^d}\bigl\lVert  \e^{\caG(x)t}(\opunit-Q_x) \bigr\rVert_{\mathrm{cb}} \leq M \e^{-gt}
\end{equation}
for all times $t\geq 0$.

As an immediate consequence, there is a (unique) state $\varrho_{\caG(x)}$ on $\caA_x$ such that
\begin{equation}
 \lim_{t\rightarrow \infty} \sigma \bigl( \e^{t \caG(x)} A \bigr)=\varrho_{\caG(x)} (A)
\end{equation}
for all initial states $\sigma$ and observables $A\in\caA_x$. In such a situation, we shall say that the semigroup $\e^{t \caG(x)}$ is relaxing to a unique stationary state.
For a volume $\La\Subset\bbZ^d$, the non-interacting dynamics on $\Lambda$ is generated by the sum of local generators 
\begin{equation}
 \caG_\La:=\sum_{x\in\La} \caG(x).
\end{equation}
The corresponding semigroup is relaxing to the unique stationary state 
\begin{equation}
 \varrho_{\caG,\La}:={\textstyle \bigotimes_{x\in\La} }\varrho_{\caG(x)}.
\end{equation}

\begin{remark} Every QMS generator $\caG$ on $\caA_x$ can be written in the Lindblad form 
\begin{equation}
 \caG(A)= i[H,A]+\sum_m K^*_m A K_m^{}-{\textstyle \frac{1}{2}\{K^*_mK_m^{}, A\}}, \qquad A\in\caA_x, 
\end{equation}
with a self-adjoint $H\in\caA_x$ and a finite family of (Kraus) operators $\{K_m\}_m$. Let $\caK$ be the algebra generated by this family. If all $\caG(x)$ are merely local copies of a generator $\caG$ as above, then the irreducibility property $\caK \psi=\caH_x$, for all $\psi\in\caH_x$, is a convenient sufficient condition for our \emph{main assumption}. See e.g.~the introductory part of \cite{JPW} for more details.  
\end{remark}

\subsubsection*{The Perturbation} The perturbation is defined by a family of local QMS generators $\{\caV(\Ga)\}$, $\Ga\Subset\bbZ^d$, $\caV(\Ga)\in \caB(\caA_\Ga)$. We assume that $\caV(\Ga)=0$ whenever $\Gamma$ is not a connected subset of $\bbZ^d$ or if $\Gamma$ is the empty set and that the norm of $\caV(\Ga)$ decays exponentially in the size of $\Ga$:
\begin{equation}\label{eq. NormInteraction}
 \lVERT  \caV \rVERT_l:=\sup_{x\in\bbZ^d} \sum_{\Ga\ni x}\e^{\lvert\Gamma\rvert/l}\,\lVert \caV(\Ga) \rVert_{\mathrm{cb}}\leq \epsilon,
\end{equation}
for some decay length $l >0$.  The constant $\epsilon$ may be viewed as the second natural energy scale  in our setup besides the decay rate $g$. 
The above condition implies that the sum of all interactions affecting a single site $x$ is bounded in norm by $\epsilon/\e^{1/l}$. The proof of our main theorem indeed does not require that each local term $\caV(\Gamma)$ generates a QMS, but only that it annihilates the identity, and that globally
\begin{equation}
 \caV_\La:=\sum_{\Ga\subset\La}\caV(\Ga)
\end{equation}
is a QMS generator for every volume $\Lambda$.

\subsubsection*{Boundary Conditions} Boundary conditions are implemented by a family of QMS generators $\{\caW(\Ga)\}$, $\Ga\Subset\bbZ^d$, $\caW(\Ga)\in \caB(\caA_\Ga)$. Again it is assumed that $\caW(\Ga)=0$ whenever $\Gamma$ is not a connected subset of $\bbZ^d$ or if $\Gamma$ is the empty set and that it satisfies {$\lVERT  \caW \rVERT_b<\infty$} for some (not necessarily small) decay length {$b>0$}. 

\section{Results}

\subsection{Open Boundary Conditions}

The weakly interacting QMS in the volume $\La\Subset\bbZ^d$ and open boundary conditions is generated by
\begin{equation}\label{eq: Generator}
 \caL_\La:=\caG_\La + \caV_\La.
\end{equation}
Let $\La\nearrow\bbZ^d$ stand for the thermodynamic limit along any sequence of finite volumes $\La$ such that every $\Delta\Subset\bbZ^d$ is a subset of almost all of them. 

\begin{theorem}\label{thm}
Assume $1/l>\log 2$ and define another decay length 
 \begin{equation}\label{eq: DecayLength}
  l'=\frac{1}{1/l -\log 2}.
 \end{equation}
 Assume further that $\epsilon < g$ and choose a decay rate $g'>0$ such that  $\epsilon< g-g'$. {We also abbreviate
 \begin{equation}\label{eq: Konst}
   K= \frac{g-g'}{g-g'-\epsilon},
 \end{equation}
 then the following holds for some constant $C>0$ that may only depend on $M$ defined in equation \eqref{eq: AssDecay}: }
 \begin{enumerate}
 \item there is a unique strongly continuous quantum Markov semigroup $\caS_t$ on the quasi-local algebra $\caA$, such that
 \begin{equation}\label{eq: ExistenceDyn}
  \bigl\lVert \e^{t\caL_{\Lambda}}(A)-\caS_t(A)\bigr\rVert\leq {(K-1)} \e^{- d(X,\Lambda^c)/l'} C^{\lvert X\rvert} \lVert A\rVert,
 \end{equation}
for all $t\geq 0$, $\Lambda\Subset \bbZ^d$, and local observables $A$ with support in $X\subset\Lambda$. 
 \item $\caS_t$ is relaxing to a unique stationary state $\varrho_{\bbZ^d}$ exponentially according to
\begin{equation}\label{eq: DecayTime}
 \bigl\lVert \caS_t(A)-\varrho_{\bbZ^d}(A)\textnormal{\opunit} \bigr\rVert\leq {K} \e^{-g' t} C^{\lvert X\rvert} \lVert A\rVert.
\end{equation}
\item  $\e^{t\caL_\Lambda}$ is relaxing to a unique stationary state $\varrho_\Lambda$, for every finite volume $\Lambda$, and furthermore
\begin{equation}
 \varrho_{\bbZ^d}=\lim_{\La\nearrow\bbZ^d}\varrho_\Lambda
\end{equation}
is the unique weak* thermodynamic limit.
\item Truncated correlations decay exponentially in space for the stationary state:
 \begin{equation}\label{eq: ExpDecayCorr}
  \bigr\rvert\varrho_{\bbZ^d}(AB)-\varrho_{\bbZ^d}(A)\varrho_{\bbZ^d}(B)\bigr\rvert\leq {(K-1)} \e^{ - d(X,Y)/l'} C^{\lvert X\rvert+\lvert Y\rvert} \lVert A\rVert\lVert B\rVert
 \end{equation}
for all pairs of local observables $A,B$ with support in $X,Y\Subset\bbZ^d$, $X\cap Y=\emptyset$.
\end{enumerate}
\end{theorem}

The theorem is proven below by a basic perturbative expansion of the dynamics exploiting the uniform local relaxation of the product dynamics for the temporal bound and the exponential decay of the interaction for the spatial bound.

\begin{remark}\label{rem: FinRange}
The definitions of the time scale $1/\epsilon$ and spatial scale $l$ are obviously not independent. The situation is more transparent for finite range interactions. In that case, the constraint $1/l>\log 2$ is no longer meaningful, since $\lVERT \caV \rVERT_l$ is finite for every $l$. Our result then can be roughly summarized in the following expected form: one can take any $g'< g$ and any $l'$ provided that the interaction satisfies 
\begin{equation}\lVERT \caV\rVERT_0 \lesssim (g-g')\e^{-R^d/l'}\end{equation}
where $R$ is the range of the interaction.
\end{remark}

\begin{remark}\label{rem: Analyticity}
If the generators $\caG(x)$ and $\caV(\Gamma)$ depend analytically on some parameters without violating the assumptions throughout an (open complex) parameter domain, then the expansions used in the proof converge uniformly. As a consequence, also the dynamics and stationary state, i.e.~$\caS_t(A)$ and $\varrho_{\bbZ^d}(A)$, for local observables $A$, depend analytically on these parameters. More precisely, this holds if the operators $\caG(x)$ and $\caV(\Gamma)$ uniformly satisfy \eqref{eq: AssDecay} and \eqref{eq. NormInteraction} and as long as they annihilate the identity $\opunit$. For a non-trivial analytic dependence on the parameters, note that $\e^{t\caL_\Lambda }$ and $\varrho_{\Lambda}$ cannot be positive for the whole domain. { Furthermore, the theorem's decay estimates remain valid throughout the parameter domain, which we however can only prove for decay lengths $l'$ that are larger than \eqref{eq: DecayLength}. The proof of the theorem is given for this more general setting and shows that $1/(a_1 l)-a_2$, for some $a_1,a_2>1$, is a lower bound on $1/l'$. }
\end{remark}

Let us also comment on the relation with some interesting recent work on extended quantum Markov dynamics. In \cite{CLMP}, the authors do not only study perturbations of products as we do, but more generally of (locally generated) dynamics $\e^{t\caL_\Lambda}$ which satisfy `rapid mixing'. In case of a unique stationary state $\varrho_\Lambda$ for each volume $\Lambda$, this condition means that
\begin{equation}
 \bigl\lVert \e^{t\caL_\Lambda} (A) - \varrho_\Lambda (A) \opunit \bigr\rVert \leq C(\lvert \Lambda \rvert) \,\e^{-g t}
\end{equation}
holds for all observables $A\in\caA_\Lambda$ and with a function $C(\lvert \Lambda \rvert)$ growing only polynomially in the volume. As a main result, it is shown in \cite{CLMP} that rapidly mixing systems are stable according to
\begin{equation}
  \bigl\lVert \e^{t\caL_\Lambda} (A) -  \e^{t\caL_\Lambda'} (A) \opunit \bigr\rVert \leq C(\lvert X \rvert) \bigl(\epsilon +b(d(X,\Lambda))\bigr)
\end{equation}
for all local observables with support $X\subset \Lambda$, where $\caL_\Lambda'$ is the perturbed generator and $b(d)\geq 0$ a decaying boundary correction. On the way to this result they also prove that rapid mixing implies the existence of the thermodynamic limit of stationary states. Note however that it is unclear and an interesting question whether perturbations of product dynamics as studied here satisfy rapid mixing. It does not follow from our result as we only obtain estimates with prefactors that grow exponentially in the support of local observables. This question is interesting also because it has been established that rapid mixing furthermore implies exponential decay of correlations and an area law (for the `mutual information') for the stationary state, see \cite{BCLMP, CLMP, KE}.\\

\subsection{Other Boundary Conditions}

The result on the existence of the infinite volume dynamics $\caS_t$ is already contained in the work of Nachtergaele et al.~\cite{NVZ} applying Lieb--Robinson propagation bounds \cite{LR} for QMS; see also \cite{H,P} for similar results in the context of (classical) Markov Processes. In fact, we use their result to prove, as a corollary of the theorem, independence from boundary conditions for the thermodynamic limit of stationary states.
For that purpose we define the following QMS generator in volume $\Lambda'\Subset\bbZ^d$ with boundary conditions outside a bulk volume $\Lambda\subset\Lambda'$,
\begin{equation}
 \caL_{\Lambda,\Lambda'}:=\caL_{\Lambda'}+\sum_{\Gamma\subset\Lambda'\setminus\Lambda}\caW(\Gamma).
\end{equation}

\begin{corollary}\label{cor}
  Under the conditions of the theorem, there are constants ${C},\tilde{l}>0$, such that the following holds for all $\Lambda\subset\Lambda'\Subset\bbZ^d$: let $\sigma$ be a stationary state for $\e^{t\caL_{\Lambda,\Lambda'}}$, then
 \begin{equation}
   \bigl\lvert \sigma(A)-\varrho_{\bbZ^d}(A) \bigr\rvert  \leq \e^{-d(X,\Lambda^c)/\tilde{l}} {C}^{\lvert X\rvert}\lVert A\rVert
 \end{equation}
 for all local observables $A$ with support in $X\subset\Lambda$, where the state $\varrho_{\bbZ^d}$ is defined as in the theorem.
\end{corollary}

This result may be viewed as a dissipative version of Yarotsky's result on ground states for weak perturbations of gapped non-interacting systems \cite{Y2}.
Here the QMS was defined by sums of local generators, and by definition each of them has a non-empty kernel containing the identity. This fact simplifies the problem considerably and is somewhat comparable with frustration-freeness when studying ground states of local Hamiltonians. Note however we are here, in case of non-equilibrium, dealing with generators that may not be self-adjoint.  \\

\subsection{Examples}
\subsubsection*{Quantum Ising Model With Dissipation} For this example we consider a chain of qubits, i.e.~each $\caH_x$, $x\in \bbZ$, is a copy of $\caH\cong \bbC^2$. We denote with $\sigma^i$, $i=1,2,3$, the usual Pauli matrices, and with $\lvert \uparrow\rangle$ and $\lvert \downarrow\rangle$ the spin up and spin down eigenstate of $\sigma^3$ in Dirac's notation.  Additionally, we put $2\sigma^{\pm}=\sigma^1\pm i \sigma^2$ such that $\sigma^{+}\lvert \downarrow\rangle=\lvert \uparrow\rangle, \sigma^{-}\lvert \uparrow\rangle=\lvert \downarrow\rangle $. 
The unperturbed product dynamics is generated by copies of a single-site generator 
\begin{equation}
 \caG(A)= i\bigl[h \sigma^3,A\bigr] + \sigma^+ A \sigma^- - {\textstyle \frac{1}{2}} \bigl\{ A, \sigma^+\sigma^-\bigr\}, \quad A\in\caB(\caH),
\end{equation}
where $h\in \bbR$ is a parameter indicating the strength of a (transverse) field, and where the expression in square brackets is the commutator and in curly brackets the anti-commutator. The coherent and dissipative part of the generator commute with each other. It gives rise to a unique stationary state $\varrho_{\caG(x)}=\lvert \downarrow \rangle \hspace{-1pt}\langle \downarrow \rvert$ at each site and \eqref{eq: AssDecay} holds $h$-uniformly for the projection $Q_x(\,\cdot\,)=\langle \downarrow \rvert \cdot\lvert \downarrow \rangle \opunit$, for example, with constants $M=4$ and $g=1/2$.  We take a translation invariant perturbation acting only on pairs of nearest neighbour sites defined through
\begin{equation}
 \caV(\{x,x+1\})(A)= i \bigl[ J \sigma_x^1\sigma_{x+1}^1 , A\bigr], \quad A\in\caA_{\{x,x+1\}}
\end{equation}
where $\sigma_x^1$ is the local operator acting on site $x$ as $\sigma^1$ and where the coupling strength $J\in \bbR$ is another parameter. The coherent part of the dynamics is that of the transverse field Ising model with Hamiltonian
\begin{equation}\label{eq: IsingHamiltonian}
 H = h\sum_{x\in \Lambda} \sigma_x^3 +J \sum_{\{x,x+1\}} \sigma_x^1\sigma_{x+1}^1.
\end{equation}
By our theorem, if the coupling $J$ is small enough (e.g.~$2J  <\e^{-2 \log 2}g$, see also Remark \ref{rem: FinRange} on finite range interactions) and for all values of $h$, the stationary state for the full dynamics generated by $\caL_\Lambda$ as in \eqref{eq: Generator} is unique even in the thermodynamic limit. We chose the particular form of perturbation above to relate to the well-known (Hamiltonian) transverse Ising model, which is exactly solvable. See also \cite{Pr}, where similar dissipative models are solved exactly.  This property is however not relevant here: uniqueness holds for any perturbation that is small enough according to the theorem. We may, for example, add a small longitudinal field $\propto \sigma^1$ at each site, then the model is not longer mapped to free fermions by a Jordan--Wigner transformation.

An immediate interesting question would be whether uniqueness persists as $J$ is increased or if there is a transition to a phase with multiple stationary states. It seems tempting (though a-priori not justified) to make a comparison with ground state properties of the Ising model in transverse field, for which indeed a quantum phase transition (in the ground state) occurs at a certain critical value of $J/h$ and there is phase coexistence above the critical value.  
This is also very much connected to a question that is not even fully resolved in the classical domain:  is it possible to devise a  one-dimensional cellular automaton with `non-degenerate' noise that can have phase coexistence, see \cite{gacs,gray}.  

\subsubsection*{Transport for Weakly Coupled Self-consistent Heat Baths}  
In our second example we start out with a product dynamics describing a quantum spin chain where each site $x\in\bbZ$ is in contact with a separate heat bath at temperature $T_x$. Given a local Hamiltonian $H_x^{}=H_x^*\in\caA_x^{}$, this contact is modelled by a QMS with unique stationary state $\varrho_{\caG(x)}=Z^{-1}\exp({-H_x/T_x})$, the Gibbs state. Its generator $\caG(x)=\caG(x,T_x)$ is detailed balance w.r.t.\ $H_x^{}$ at temperature $T_x>0$, i.e.\ it is self-adjoint 
for the inner product `weighted' with the Gibbs state $\varrho_{\caG(x)}$, 
\begin{equation}
 \langle A,B\rangle_{x}=\Tr\bigl( A^* \varrho_{\caG(x)}^s B \varrho_{\caG(x)}^{1-s} \bigr).
\end{equation}
for some $0<s<1$ and it depends analytically on $T_x$, see Remark \ref{rem: Analyticity}.  The perturbation $\caV$ is a nearest neighbour interaction of the form
\begin{equation} \label{eq: continuity equation}
 \caV (\{x,x+1\})=i [V_{x,x+1}, \cdot ], \quad \text{such that} \quad [V_{x,x+1}, H_x+H_{x+1}]=0
\end{equation}
 supposed to be sufficiently weak, so that the theorem guarantees a unique stationary state $\varrho_\Lambda=\varrho_\Lambda(T_1,\dots,T_N)$, for all volumes $\Lambda=\{1,\dots,N\}\subset \bbZ$. 
Using that the interaction conserves energy locally, i.e.\ \eqref{eq: continuity equation}, we can write the continuity equation
\begin{align}
 0&=\partial_t \varrho_\Lambda^{}\bigl( \e^{t\caL_\Lambda} (H_x) \bigr) \\
 &=\varrho_\Lambda^{}\bigl( \caV(\{x-1,x\}) (H_x) \bigr)-\varrho_\Lambda^{}\bigl( \caV(\{x,x+1\}) (H_{x+1}) \bigr) + \varrho_\Lambda^{}\bigl( \caG(x) (H_x) \bigr)\nonumber\\
 &\equiv j_x(T_1,\dots,T_N)-j_{x+1}(T_1,\dots,T_N)+J_x(T_1,\dots,T_N)\nonumber
\end{align}
where $j_x$ is the current from site $x-1$ to $x$ and $J_x$ the current originating from the heat bath at site $x$. The game here is the following: How should one choose the temperatures $T_j$ such that the currents $\boldsymbol{J}=J_2,\dots,J_{N-1}$ from the heat baths in the bulk all vanish (but not necessarily $J_1$ and $J_N$).  Of course, choosing all $T_j$ equal is a solution (corresponding to global equilibrium).  But there are other solutions,  e.g.  corresponding to a uniform heat current $j_{\mathrm{sc}}$ through the chain,
\begin{equation}
 J_1=j_x=-J_N=j_{\mathrm{sc}}, \quad x=2, \dots, N.
\end{equation}
The idea motivating such a procedure is that the coupled heat baths model chaotic degrees of freedom along the chain, so that indeed the steady currents from these reservoirs should vanish. 
The non-trivial question is now: How does the profile of $T_j$ look like for a nonzero $j_{\mathrm{sc}}$.  The problem can be solved for sufficiently small $T_N-T_1$, and the result is that the current $j_{\mathrm{sc}}$ and local temperature differences $T_{x+1}-T_x$ both scale as $1/N$ obeying Fourier's law $j_{\mathrm{sc}}=\kappa(T_x)(T_{x+1}-T_x)$ for some heat conductivity $\kappa$, see \cite{BLR}.
We do not prove this theorem here as it strays too far from the main message of this paper. However, some thought shows that the result of the present paper allows to establish this, as it shows that the local currents depend analytically and exponentially weak on distant temperature parameters. For related results on such chains of `self-consistent heat baths' in classical mechanics, see e.g.~\cite{lukkarinenself}.

\section{Proofs}

\subsection{Proof of the Theorem} For every volume $\Lambda \Subset \bbZ^d$ the QMS can be expressed in the form of the norm-convergent Dyson expansion. For any local observable $A$ with support in $X\subset\Lambda$, $X\neq \emptyset$, one obtains
\begin{equation}
\begin{split}
  \e^{t\caL_\Lambda}(A) =\e^{t\caG_\Lambda}(A)+\sum_{n=1}^\infty \int_{t_1\leq\dots\leq t_n\leq t} \hspace{-8ex}\mathrm{d}t_1\dots \mathrm{d}t_n\; \e^{(t-t_n)\caG_\Lambda}& \caV_\Lambda \e^{(t_n-t_{n-1})\caG_\Lambda} \caV_\Lambda\dots\\
  &\dots\e^{(t_2-t_1)\caG_\Lambda} \caV_\Lambda \e^{t_1\caG_\Lambda} (A).
\end{split}
\end{equation}
Given a subset $E\subset \La$ we define the projection operator
\begin{equation}
 P_\Lambda(E):=\bigl({\textstyle \bigotimes_{x\in E}} (\opunit-Q_x)\bigr) \otimes \bigl({\textstyle \bigotimes_{x\in \La\setminus E}} Q_x\bigr) 
\end{equation}
on $\caA_{\La}$. Since the one-dimensional range of each $Q_x$ is spanned by the identity we find that
\begin{equation}
 \caV(\Ga)P_\Lambda(E)=0\quad \text{if}\quad \Ga\cap E=\emptyset
\end{equation}
and in particular $\caV(\Ga)P_\Lambda(\Lambda)=0$ for every $\Gamma\subset\Lambda$. The locality of the expansion and the local relaxation assumption \eqref{eq: AssDecay} become accessible by writing out the sum of local terms in $\caV_\Lambda$ and inserting the identity in the form 
\begin{equation}
 \opunit=\sum_{E\subset \La}P_\Lambda(E)
\end{equation}
multiple times in the Dyson expansion. It can be written as
\begin{equation}\label{eq: DysonExtended}
 \e^{t\caL_\Lambda}(A) =\e^{t\caG_\Lambda}(A)+\sideset{}{^\Lambda}\sum_{\boldsymbol{\Gamma}_n,\boldsymbol{E}_n}\int_0^t\mathrm{d}s \, H_t(s,\boldsymbol{\Gamma}_n,\boldsymbol{E}_n,A),
\end{equation}
where we abbreviated $\boldsymbol{\Ga}_n=(\Ga_1,\dots,\Ga_n)$, $\bsE_n=(E_1,\dots,E_{n+1})$, with non-empty $\Ga_i, E_i \subset \La$ for all $i=1,\dots, n$, and $E_{n+1}\subset\Lambda$ also allowed to be the empty set. The above sum runs over all these $n$-tupels, $n\geq 1$, of subsets of the volume $\Lambda$. Furthermore we introduced the operator-valued function
\begin{align}\label{eq: DefinitionH}
 &H_t(s,\boldsymbol{\Gamma}_n,\boldsymbol{E}_n,A)\\
 &:=\int_{t_1\leq\dots\leq t_{n-1}\leq s} \hspace{-9ex}\mathrm{d}t_1\dots \mathrm{d}t_{n-1}\;  \e^{(t-s)\caG_\Lambda}P_\Lambda(E_{n+1})\bigl( \caV(\Gamma_n) \e^{(s-t_{n-1})\caG_\Lambda} P_\Lambda(E_n)\bigr)\dots \notag\\
 &\hphantom{:=\int_{t_1\leq\dots\leq t_{n-1}\leq s} \hspace{-9ex}\mathrm{d}t_1\dots \mathrm{d}t_{n-1}\; P_\Lambda(E_{n+1}) \e^{(t-s)\caG_\Lambda}}\dots\bigl(\caV(\Gamma_1) \e^{t_1\caG_\Lambda} P_\Lambda(E_1)\bigr) \,(A).\notag
 \end{align}
For given $\boldsymbol{\Ga}_n$ and $\bsE_n$ let us define the set
\begin{equation}
\begin{split}
 &D=D(\boldsymbol{\Ga}_n):={\textstyle \bigcup_{i=1}^n}\Ga_i
\end{split}
\end{equation}
and we state some rather direct observations regarding $H_t(s,\boldsymbol{\Gamma}_n,\boldsymbol{E}_n,A)$: 
\begin{enumerate}[(i)]
 \item $H_t$ vanishes unless
 \begin{equation}\label{eq: QueueRule}
 E_1\subset X, \quad E_{i+1} \setminus \Ga_{i}= E_{i} \setminus \Ga_{i}\quad \text{and} \quad E_i \cap \Ga_i \neq \emptyset, \quad 1\leq i \leq n
\end{equation}
where $X$ is the support of $A$.
\item $H_t$ does not depend on the volume $\Lambda$ as long as $X,D\subset \Lambda$.
\item $H_t$ does not depend on the time $t$ if $E_{n+1}=\emptyset$.
\item $H_t$ is a local operator whose support is within $E_{n+1}$. In particular it is proportional to the identity if $E_{n+1}=\emptyset$. \\
\end{enumerate}

The above expansion, or more precisely the integrand in \eqref{eq: DefinitionH}, can be depicted by diagrams that spread into space-time starting from $X$ at time zero and that are connected as a consequence of observation (i), see Figure \ref{fig: 1}. \\

\begin{figure}[h]
        \def\svgwidth{.7\textwidth}%
        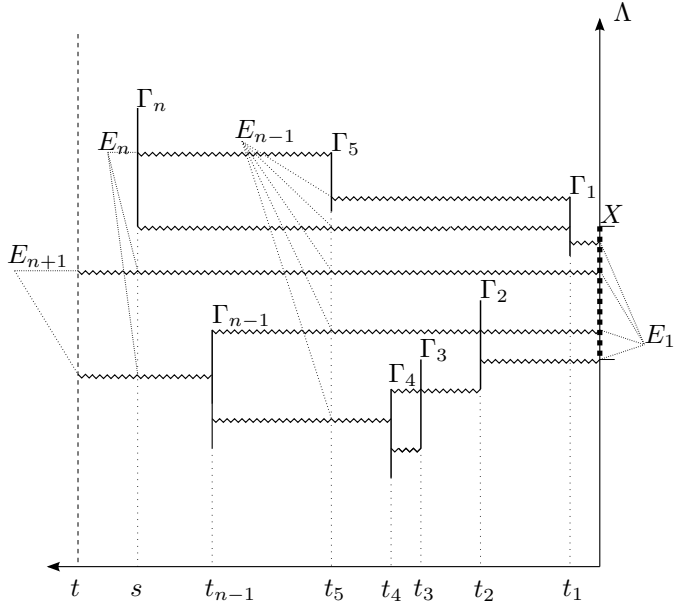
 \caption{\small (a) Sample diagram for $n=7$: the vertical lines indicate perturbations $\caV(\Gamma_i)$ acting on $\Gamma_i$ at time $t_i$, the horizontal zig-zag lines indicate the sets within $\boldsymbol{E}_n$.}\label{fig: 1}
\end{figure}

We continue the proof with two Lemmata; the first one concerns the time integral of $H_t(s,\boldsymbol{\Gamma}_n,\boldsymbol{E}_n,A)$ and then we show $\Lambda$-uniform summability over the spatial subsets $\boldsymbol{\Gamma}_n$ and $\boldsymbol{E}_n$ in the second Lemma. $\chi$ denotes the indicator function.

\begin{lemma}\label{lem: One}
 For every { $l', g'\in \bbR$, $g'<g$, there there is $l''>0$}, such that for all volumes $\Lambda\Subset\bbZ^d$, $\boldsymbol{\Gamma}_n$ and $\boldsymbol{E}_n$ within $\Lambda$, all times $t\geq 0$, and $n\geq 1$,
   \begin{align}\label{eq: TimeBound}
    &\int_{t_1}^{t_2} \mathrm{d}s\, \bigl\lVert  H_t(s,\boldsymbol{\Gamma}_n,\boldsymbol{E}_n,A) \bigr\rVert\\
    &{\leq \exp\bigl\{{-g' t \chi[E_{n+1}\neq \emptyset]}{-g' t_1\chi[E_{n+1}= \emptyset]}{-l' \lvert D\rvert}\bigr\}}\notag\\
    &\hphantom{\leq}\times C^{\lvert X\rvert}\lVert A \rVert\prod_{i=1}^n \frac{\e^{\lvert\Gamma_i\rvert/l''}\lVert \caV(\Gamma_i)\rVert_{\mathrm{cb}}}{{(g-g')}\lvert E_i\rvert}\notag
   \end{align}
   holds for all times $0\leq t_1\leq t_2 \leq t${. Here, one can take $l''=l'$ and $C=1$ under the conditions of the theorem, and $1/l''=\log (M+1) (1+1/l')$ and $C=M+1$ for the more general setting with complex parameter dependence introduced in Remark \ref{rem: Analyticity}.}\\
\end{lemma}

\begin{proof}
  { In the proof we first focus on the more general setting.}
 By observation (iii) we can, for every given $\boldsymbol{\Gamma}_n$ and $\boldsymbol{E}_n$, restrict the volume to $\Lambda= X\cup D$. Note that, for all $t\geq 0$ and $E\subset \Lambda$, 
 \begin{equation}\label{eq: FreeOnE}
  \e^{t\caG_\Lambda} P_\Lambda (E)=\bigl( {\textstyle \bigotimes_{x\in E} }\, \e^{t\caG(x)}(\opunit-Q_x) \bigr)\otimes \bigl({\textstyle \bigotimes_{x\in E^c} }\,Q_x\bigr)
 \end{equation}
 is a product of single-site operators, which are bounded in the completely bounded norm by $M\e^{-gt}$ or $M+1$ by assumption \eqref{eq: AssDecay}.  The integrand in \eqref{eq: DefinitionH} can be  bounded in norm by
 \begin{align}\label{eq: IntegrandBound}
   &\bigl\lVert  \e^{(t-s)\caG_\Lambda}P_\Lambda(E_{n+1}) \bigl( \caV(\Gamma_n) \e^{(s-t_{n-1})\caG_\Lambda} P_\Lambda(E_n)\bigr)\dots  \\
   &\hphantom{\leq\times\bigl\lVert  \bigl( \caV(\Gamma_n) \e^{(t-t_{n-1})\caG_{\Lambda\setminus\Gamma_n}}}\dots\bigl(\caV(\Gamma_1) \e^{t_1\caG_X} P_X(E_1)\bigr)\,(A) \bigl\lVert \notag\\[.3\baselineskip] 
   &\leq (M+1)^{\lvert \Gamma_n\rvert}\,\e^{-g(t-s)\lvert \Gamma_n\cap E_{n+1}\rvert}\notag \\ 
   &\hphantom{\leq}\times\bigl\lVert  \bigl( \caV(\Gamma_n) \e^{(t-t_{n-1})\caG_{\Lambda\setminus\Gamma_n}}\,\e^{(s-t_{n-1})\caG_{\Gamma_n}}  P_\Lambda(E_n)\bigr)\dots \notag \\
   &\hphantom{\leq\times\bigl\lVert  \bigl( \caV(\Gamma_n) \e^{(t-t_{n-1})\caG_{\Lambda\setminus\Gamma_n}}}\dots\bigl(\caV(\Gamma_1) \e^{t_1\caG_X} P_X(E_1)\bigr)\,(A) \bigl\lVert \notag 
 \end{align}
 Repeating this step $n$ times gives the upper bound 
 \begin{align}
   & (M+1)^{\lvert X\rvert+\sum_i\lvert \Gamma_i\rvert}\,\e^{-g(t-s)\lvert E_{n+1}\rvert}\lVert A\rVert\prod_{i=1}^n\e^{-g(t_i-t_{i-1})\lvert E_{i}\rvert} \lVert \caV(\Gamma_i)\rVert_{\mathrm{cb}}
 \end{align}
where we set $t_0:=0$ and as before $s:=t_n$. The diagrammatic representation can be useful for this exercise.
Since $\lvert D\rvert\leq\sum_i \lvert E_i\rvert$, we then get
\begin{align}\label{eq: AbsorbExp}
  &\int_{t_1}^{t_2} \mathrm{d}s\, \bigl\lVert  H_t(s,\boldsymbol{\Gamma}_n,\boldsymbol{E}_n,A) \bigr\rVert\,\exp\bigl\{{{g'} t \chi[E_{n+1}\neq \emptyset]}{+{g'} t_1\chi[E_{n+1}= \emptyset]}{+ \lvert D\rvert{/l'}}\bigr\}\\
  &\leq (M+1)^{\lvert X\rvert+(1+{1/l'})\sum_i \lvert E_i\rvert}\lVert A\rVert\notag\\
  &\hphantom{(M+1)}\times \int_{t_1}^{t_2} \mathrm{d}t_n\,\int_{t_1\leq\dots\leq t_{n-1}\leq t_n} \hspace{-9ex}\mathrm{d}t_1\dots \mathrm{d}t_{n-1}\;\prod_{i=1}^n\e^{-{(g-g')}(t_i-t_{i-1})\lvert E_{i}\rvert} \lVert \caV(\Gamma_i)\rVert_{\mathrm{cb}}.\notag
\end{align}
By changing the integration variables to $s_i:=t_i-t_{i-1}$ and extending the integration domain, we obtain another upper bound for \eqref{eq: AbsorbExp},
\begin{align}
 &(M+1)^{\lvert X\rvert+(1+{1/l'})\sum_i \lvert E_i\rvert}\lVert A\rVert \prod_{i=1}^n\int_0^\infty \mathrm{d}s_i\,\e^{-g's_i\lvert E_{i}\rvert} \lVert \caV(\Gamma_i)\rVert_{\mathrm{cb}}\\
 &=(M+1)^{\lvert X\rvert}  \lVert A\rVert \prod_{i=1}^n \frac{(M+1)^{(1+{1/l'})\lvert\Gamma_i\rvert}\lVert \caV(\Gamma_i)\rVert_{\mathrm{cb}}}{{(g-g')}\lvert E_i\rvert}.\notag
\end{align}
This finishes the proof of the lemma in the setting of Remark \ref{rem: Analyticity}.

{ Under the conditions of the theorem, we proceed just as above, but we can additionally exploit the fact that $\lVert Q_x\rVert_{\mathrm{cb}}= 1$, since $Q_x$ is the $t\rightarrow \infty$ limit of a (contractive) QMS.}
\end{proof}

\begin{lemma}
 { {Recall the definition of $K$ in \eqref{eq: Konst} of the theorem.} For every $l''>0$, 
 \begin{equation}\label{eq: lem2}
  \sideset{}{^{\bbZ^d}}\sum_{\boldsymbol{\Gamma}_n,\boldsymbol{E}_n: \mathrm{(i)}}\prod_{i=1}^n \frac{\e^{\lvert\Gamma_i\rvert/l''}\lVert \caV(\Gamma_i)\rVert_{\mathrm{cb}}}{{(g-g')}\lvert E_i\rvert}\leq {(K-1)} 2^{\lvert X\rvert}
 \end{equation} 
 if $\lVERT  \caV \rVERT_l < g-g'$ for $1/l=\log 2 +1/l''$.}\\
\end{lemma}

\begin{proof}
First recall the constraint \eqref{eq: QueueRule} in observation (i). Fixing $E_i$ and $\Gamma_i$ for some $1\leq i\leq n$ fully determines a family of at most $2^{\lvert \Gamma_i\rvert}$ different compatible subsets $E_{i+1}$. On the other hand, fixing $E_i$ determines a family of compatible subsets $\Gamma_i$. Then, for every $E_i\Subset\bbZ^d$,
\begin{align}
 &\sideset{}{'}\sum_{\Ga_i}\sideset{}{'}\sum_{E_{i+1}} \frac{\e^{\lvert\Gamma_i\rvert{ /l''}}\lVert \caV(\Gamma_i)\rVert_{\mathrm{cb}}}{{(g-g')}\lvert E_i\rvert}\\
 &\leq  \sum_{x\in E_i} \sum_{\Ga\ni x} \frac{2^{\lvert\Gamma_i\rvert}\e^{\lvert\Gamma_i\rvert{ /l''}}\lVert \caV(\Gamma_i)\rVert_{\mathrm{cb}}}{{(g-g')}\lvert E_i\rvert}\notag\\
 &\leq \frac{\lVERT\caV \rVERT_{l}}{{g-g'}}\quad \text{with}\quad 1/l=\log 2+ 1/l''\notag
\end{align}
 where we put primes to indicate that we sum over compatible subsets depending on $E_i$ and $E_i,\Gamma_i$ respectively.
 Therefore we obtain the majorizing geometric series
 \begin{align}
   &\sideset{}{^{\bbZ^d}}\sum_{\boldsymbol{\Gamma}_n,\boldsymbol{E}_n: \mathrm{(i)}}\prod_{i=1}^n \frac{\e^{\lvert\Gamma_i\rvert{ /l''}}\lVert \caV(\Gamma_i)\rVert_{\mathrm{cb}}}{{(g-g')}\lvert E_i\rvert}\\
   &=\sum_{n=1}^\infty \sum_{E_1\subset X} \sideset{}{'}\sum_{\Ga_1}\sideset{}{'}\sum_{E_{2}}\frac{\e^{\lvert\Gamma_1\rvert{ /l''}}\lVert \caV(\Gamma_1)\rVert_{\mathrm{cb}}}{{(g-g')}\lvert E_1\rvert}\dots\sideset{}{'}\sum_{\Ga_n}\sideset{}{'}\sum_{E_{n+1}}\frac{\e^{\lvert\Gamma_n\rvert{ /l''}}\lVert \caV(\Gamma_n)\rVert_{\mathrm{cb}}}{{(g-g')}\lvert E_n\rvert}\notag\\
   &\leq 2^{\lvert X\rvert} \sum_{n=1}^\infty \,\Bigl(\frac{\lVERT\caV \rVERT_{l}}{{g-g'}}\Bigr)^n,\notag
  \end{align}
if $\lVERT  \caV \rVERT_l < g-g'$, which proves the Lemma.\\
\end{proof}

The above Lemmata together show that the expansion \eqref{eq: DysonExtended} converges in norm, even if weighted exponentially,
\begin{align}\label{eq: ExpansionExpDec}
 &\sideset{}{^{\bbZ^d}}\sum_{\boldsymbol{\Gamma}_n,\boldsymbol{E}_n}\int_{t_1}^{t_2}\mathrm{d}s \,\bigl\lVert H_t(s,\boldsymbol{\Gamma}_n,\boldsymbol{E}_n,A)\bigr\rVert\\
 &\hphantom{\sideset{}{^{\bbZ^d}}\sum_{\boldsymbol{\Gamma}_n,\boldsymbol{E}_n}\int_{t_1}^{t_2}}\times\exp\bigl\{{{g' t}\chi[E_{n+1}\neq \emptyset]   }{+{g'} t_1\chi[E_{n+1}= \emptyset]}{+ \lvert D\rvert{/l'}}\bigr\}\notag\\
 &\leq {(K-1)} C^{\lvert X \rvert} \lVert A\rVert\notag
\end{align}
for all times $0\leq t_1,t_2\leq t$ { and with some constant $C>0$ { (that only depends on $M$)}, if $\epsilon=\lVERT  \caV \rVERT_l < g-g'$,  where either
\begin{equation}
 \begin{split}
  &1/l=\log 2 +1/l' \quad \text{or}\\
  &1/l=\log 2 +\log(M+1)(1+1/l') 
 \end{split}
\end{equation}
in the setting of the theorem or of Remark \ref{rem: Analyticity} respectively.
}

In the following, the constant $C$ represents a family of constants and its value may change from line to line. Let us describe further consequences of the Lemmata:

\noindent
(1)  The thermodynamic limit of the dynamics
\begin{align}\label{eq: DefS}
 \caS_t(A):&=\lim_{\Lambda\nearrow \bbZ^d}\e^{t\caL_\Lambda}(A)\\
 &=\e^{t\caG_X}(A)+\sideset{}{^{\bbZ^d}}\sum_{\boldsymbol{\Gamma}_n,\boldsymbol{E}_n}\int_0^t\mathrm{d}s \, H_t(s,\boldsymbol{\Gamma}_n,\boldsymbol{E}_n,A)\notag
\end{align}
exists for each local observable $A\in\caA_X$, $X\Subset \bbZ^d$, and $t\geq 0$. Since the local observables are dense in $\caA$, the infinite volume quantum Markov semigroup $\caS_t$ of the theorem can be defined by continuous extension. Let $\Lambda,\Lambda' \Subset \bbZ^d$ with $\Lambda\subset \Lambda'$. In the difference
\begin{align}
  &\e^{t\caL_\Lambda}(A)-\e^{t'\caL_{\Lambda'}}(A)\\
  &= \sideset{}{^{\Lambda'}}\sum_{\boldsymbol{\Gamma}_n,\boldsymbol{E}_n} \chi[D\cap\Lambda^c\neq \emptyset]\int_0^t \mathrm{d}s\, H_t(s,\boldsymbol{\Gamma}_n,\boldsymbol{E}_n,A)
 \end{align}
all those terms that, diagrammatically speaking, do not reach outside of $\Lambda$ cancel by observation (ii). Moreover the remaining terms in the sum decay exponentially in the distance $d(X,\Lambda^c)\leq\lvert D\rvert$,
\begin{equation}
\bigl\lVert \e^{t\caL_\Lambda}(A)-\e^{t'\caL_{\Lambda'}}(A)\bigr\rVert \leq { (K-1)}\e^{- d(X,\Lambda^c){/l'}} C^{\lvert X\rvert} \lVert A\rVert.
\end{equation}
This shows that \eqref{eq: ExistenceDyn} holds by taking the limit $\Lambda'\nearrow \bbZ^d$ and that  the convergence of \eqref{eq: DefS} is uniform for all times. Therefore $\caS_t$ is strongly continuous.

\hspace{5pt}

\noindent
(2) We take $\varrho_{\bbZ^d}$ to be the state on $\caA$ defined by
\begin{equation}
 \varrho_{\bbZ^d}(A)\opunit= \varrho_{\caG,X}(A)\opunit+\lim_{t\rightarrow \infty} \, \sideset{}{^{\bbZ^d}}\sum_{\boldsymbol{\Gamma}_n,\boldsymbol{E}_n} \chi[E_{n+1}=\emptyset]\int_0^t\mathrm{d}s \, H_t(s,\boldsymbol{\Gamma}_n,\boldsymbol{E}_n,A)\\
\end{equation}
for local observables and extend it by continuity to the full algebra of observables.  We find that
\begin{align}
 &\caS_t(A)-\varrho_{\bbZ^d}(A)\opunit\\
 &= \e^{t \caG_X}(A)-\varrho_{\caG,X}(A)\opunit\notag\\
 & +\sideset{}{^{\bbZ^d}}\sum_{\boldsymbol{\Gamma}_n,\boldsymbol{E}_n}\chi[E_{n+1}\neq\emptyset]\int_0^{t}\mathrm{d}s \, H_{t}(s,\boldsymbol{\Gamma}_n,\boldsymbol{E}_n,A)\notag\\
 & -\lim_{t'\rightarrow \infty} \,\sideset{}{^{\bbZ^d}}\sum_{\boldsymbol{\Gamma}_n,\boldsymbol{E}_n}\chi[E_{n+1}=\emptyset]\int_t^{t'}\mathrm{d}s \, H_{t'}(s,\boldsymbol{\Gamma}_n,\boldsymbol{E}_n,A).\notag
 \end{align}
Due to the fact that the unperturbed dynamics is relaxing exponentially fast, the first term is bounded by
\begin{align}
 \bigl\lVert\e^{t \caG_X}(A)-\varrho_{\caG,X}(A)\opunit\bigr\rVert&=\Bigl\lVert \sum_{\emptyset\neq E\subset X} \bigr(\e^{t \caG_X}P_X(E)\bigl)(A)\Bigr\rVert\\
 &\leq 2^{\lvert X \rvert} (M+1)^{\lvert X\rvert}\, \e^{-g t }\lVert A \rVert. \notag
\end{align}
This bound only uses \eqref{eq: AssDecay}, which is convenient in view of the remark \ref{rem: Analyticity} on analytic dependence. For products of QMS one can easily improve the above bound to obtain a prefactor that grows only linearly in $\lvert X\rvert$. By \eqref{eq: ExpansionExpDec}, the other two terms 
{ are each bounded by 
\begin{equation}
 (K-1)\e^{-g' t}C^{\lvert X \rvert}\lVert A\rVert    ,
\end{equation}
 so that the theorem's claim \eqref{eq: DecayTime} follows for some appropriate constant $C$.
}

\hspace{5pt} 

\noindent(3) Defining states $\varrho_\Lambda$ as above, but restricting to finite volumes $\Lambda$, it also follows in the same way that $\e^{t\caL_\Lambda}$ is relaxing to a unique stationary states exponentially fast,
\begin{equation}\label{eq: ExpDecayFiniteVol}
 \bigl\lVert \e^{t\caL_\Lambda}(A)-\varrho_{\Lambda}(A)\opunit\bigr\rVert\leq {K} \e^{-{g'} t} C^{\lvert X\rvert} \lVert A\rVert.
\end{equation}
Furthermore, we have
\begin{equation}\label{eq: DecayState}
 \bigl\lvert \varrho_\Lambda(A)-\varrho_{\bbZ^d}(A) \bigr\rvert\leq {(K-1)}\e^{- d(X,\Lambda^c){/l'}} C^{\lvert X\rvert} \lVert A\rVert.
\end{equation}
It follows that $\varrho_{\bbZ^d}$ is the (unique) weak* limit of these stationary states.

\hspace{5pt} 

\noindent (4) Exponential decay of correlations follows from the expansion and \eqref{eq: ExpansionExpDec}, since all terms with $\lvert D\rvert \leq d(X,Y)$ cancel in the difference \eqref{eq: ExpDecayCorr}.

\subsection{Proof of the Corollary}
By assumption, the local generators of the bulk and boundary dynamics $\caV$ and $\caW$ decay exponentially in the diameter of their support (assumed here to be connected sets). This infers the existence of an effective `propagation speed' in the system, see \cite{NVZ}: adapted to our purposes we may conclude that there exist constants $C,\mu,v >0$, such that
\begin{equation}\label{eq: LrBound}
 \bigl\lVert \e^{t\caL_{\Lambda,\Lambda'}}(A)- \e^{t\caL_{\Lambda}}(A)\bigr\rVert\leq \e^{-\mu (d(X,\Lambda^c)-v t)} C^{\lvert X\rvert} \lVert A\rVert
\end{equation}
for all local observables $A$ with support in $X\subset \Lambda$ and times $t\geq 0$. [To facilitate transferring the results from \cite{NVZ} note that, in case of the $d$-dimensional lattice, one may take $\e^{-\mu n}(1+n)^{-(1+d)}$ for the function $F_\mu(n)$ appearing in this reference.] Inserting the time evolution appropriately and using the triangle inequality as in \cite{CLMP} we obtain
\begin{align}
 &\bigl\lvert \sigma(A)-\varrho_{\Lambda}(A)\bigr\rvert\\
 &\leq \bigl\lvert \sigma\bigl(\e^{t\caL_{\Lambda,\Lambda'}}(A)\bigr)- \sigma\bigl(\e^{t\caL_{\Lambda}}(A)\bigr)\bigr\rvert+\bigl\lvert \sigma\bigl(\e^{t\caL_{\Lambda}}(A)-\sigma\bigl(\varrho_{\Lambda}(A)\opunit\bigr)\bigr\rvert\notag\\
 &\leq \bigl\lVert \e^{t\caL_{\Lambda,\Lambda'}}(A)- \e^{t\caL_{\Lambda}}(A)\bigr\rVert+\bigl\lVert \e^{t\caL_{\Lambda}}(A)-\varrho_{\Lambda}(A)\opunit\bigr\rVert\notag
\end{align}
for arbitrary $t\geq 0$. Choosing this time through $vt=d(X,\Lambda^c)/2$ hence gives
\begin{equation}
 \bigl\lvert \sigma(A)-\varrho_{\Lambda}(A)\bigr\rvert\leq \e^{-d(X,\Lambda^c)/\tilde{l}} C^{\lvert X\rvert} \lVert A\rVert, \quad {\tilde{l}:=\max\{l',2/\mu\}}
\end{equation}
by the propagation bound \eqref{eq: LrBound} and exponentially fast relaxation in finite volume \eqref{eq: ExpDecayFiniteVol}. Together with \eqref{eq: DecayState} we arrive at the Corollary.

\bibliographystyle{plain}

\end{document}

%% file: drawing1.pdf_tex
%% Creator: Inkscape inkscape 0.48.4, www.inkscape.org
%% PDF/EPS/PS + LaTeX output extension by Johan Engelen, 2010
%% Accompanies image file 'drawing1.pdf' (pdf, eps, ps)
%%
%% To include the image in your LaTeX document, write
%%   \input{<filename>.pdf_tex}
%%  instead of
%%   \includegraphics{<filename>.pdf}
%% To scale the image, write
%%   \def\svgwidth{<desired width>}
%%   \input{<filename>.pdf_tex}
%%  instead of
%%   \includegraphics[width=<desired width>]{<filename>.pdf}
%%
%% Images with a different path to the parent latex file can
%% be accessed with the `import' package (which may need to be
%% installed) using
%%   \usepackage{import}
%% in the preamble, and then including the image with
%%   \import{<path to file>}{<filename>.pdf_tex}
%% Alternatively, one can specify
%%   \graphicspath{{<path to file>/}}
%% 
%% For more information, please see info/svg-inkscape on CTAN:
%%   http://tug.ctan.org/tex-archive/info/svg-inkscape
%%
\begingroup%
  \makeatletter%
  \providecommand\color[2][]{%
    \errmessage{(Inkscape) Color is used for the text in Inkscape, but the package 'color.sty' is not loaded}%
    \renewcommand\color[2][]{}%
  }%
  \providecommand\transparent[1]{%
    \errmessage{(Inkscape) Transparency is used (non-zero) for the text in Inkscape, but the package 'transparent.sty' is not loaded}%
    \renewcommand\transparent[1]{}%
  }%
  \providecommand\rotatebox[2]{#2}%
  \ifx\svgwidth\undefined%
    \setlength{\unitlength}{722.4bp}%
    \ifx\svgscale\undefined%
      \relax%
    \else%
      \setlength{\unitlength}{\unitlength * \real{\svgscale}}%
    \fi%
  \else%
    \setlength{\unitlength}{\svgwidth}%
  \fi%
  \global\let\svgwidth\undefined%
  \global\let\svgscale\undefined%
  \makeatother%
  \begin{picture}(1,0.95570321)%
    \put(0,0){\includegraphics[width=\unitlength]{drawing1.pdf}}%
    \put(0.93031166,0.90659543){\color[rgb]{0,0,0}\makebox(0,0)[lb]{\smash{$\Lambda$}}}%
    \put(0.1234773,0.04597211){\color[rgb]{0,0,0}\makebox(0,0)[lb]{\smash{$t$}}}%
    \put(0.90974529,0.60521907){\color[rgb]{0,0,0}\makebox(0,0)[lb]{\smash{$X$}}}%
    \put(0.21207087,0.04597211){\color[rgb]{0,0,0}\makebox(0,0)[lb]{\smash{$s$}}}%
    \put(0.32281285,0.04597211){\color[rgb]{0,0,0}\makebox(0,0)[lb]{\smash{$t_{n-1}$}}}%
    \put(0.5,0.04597211){\color[rgb]{0,0,0}\makebox(0,0)[lb]{\smash{$t_5$}}}%
    \put(0.5849166,0.04534486){\color[rgb]{0,0,0}\makebox(0,0)[lb]{\smash{$t_4$}}}%
    \put(0.63289037,0.04597211){\color[rgb]{0,0,0}\makebox(0,0)[lb]{\smash{$t_3$}}}%
    \put(0.72284659,0.04534486){\color[rgb]{0,0,0}\makebox(0,0)[lb]{\smash{$t_2$}}}%
    \put(0.85437431,0.04597211){\color[rgb]{0,0,0}\makebox(0,0)[lb]{\smash{$t_1$}}}%
    \put(0.22156302,0.77766015){\color[rgb]{0,0,0}\makebox(0,0)[lb]{\smash{$\Gamma_n$}}}%
    \put(0.33388704,0.45571741){\color[rgb]{0,0,0}\makebox(0,0)[lb]{\smash{$\Gamma_{n-1}$}}}%
    \put(0.86465749,0.64793384){\color[rgb]{0,0,0}\makebox(0,0)[lb]{\smash{$\Gamma_1$}}}%
    \put(0.73176712,0.49289505){\color[rgb]{0,0,0}\makebox(0,0)[lb]{\smash{$\Gamma_2$}}}%
    \put(0.64238253,0.40351045){\color[rgb]{0,0,0}\makebox(0,0)[lb]{\smash{$\Gamma_3$}}}%
    \put(0.59599079,0.36669124){\color[rgb]{0,0,0}\makebox(0,0)[lb]{\smash{$\Gamma_4$}}}%
    \put(0.5110742,0.71279698){\color[rgb]{0,0,0}\makebox(0,0)[lb]{\smash{$\Gamma_5$}}}%
    \put(0.02934662,0.54431098){\color[rgb]{0,0,0}\makebox(0,0)[lb]{\smash{$E_{n+1}$}}}%
    \put(0.97619048,0.42565884){\color[rgb]{0,0,0}\makebox(0,0)[lb]{\smash{$E_1$}}}%
    \put(0.16302801,0.713588){\color[rgb]{0,0,0}\makebox(0,0)[lb]{\smash{$E_n$}}}%
    \put(0.36710963,0.73257233){\color[rgb]{0,0,0}\makebox(0,0)[lb]{\smash{$E_{n-1}$}}}%
  \end{picture}%
\endgroup%